\documentclass[12pt,oneside]{amsart}

\usepackage{amsaddr}
\usepackage{amsthm, amsmath, amssymb, amsbsy, amscd, amsfonts,url,
}

\allowdisplaybreaks
\newtheorem{proposition}{Proposition}
\newtheorem{corollary}{Corollary}

\newtheorem{theorem}{Theorem}

\theoremstyle{definition}
\newtheorem{definition}{Definition}
\newtheorem{example}{Example}
\newtheorem{remark}{Remark}

\newcommand{\md}{W}
\newcommand{\Mfd}{\mathbf{M}}
\newcommand{\cl}{\colon}
\newcommand{\er}{\eqref}
\newcommand{\lb}{\label}

\newcommand{\qqquad}{\qquad\quad}

\newcommand{\tm}{\mathbb{T}}
\newcommand{\tbf}{\mathbf{T}}

\newcommand{\ftm}{\mathbb{T}}
\newcommand{\stm}{\mathbf{T}}

\newcommand{\zsp}{\mathbb{Z}_{>0}}

\DeclareMathOperator{\mat}{Mat}

\newcommand{\xc}{\mathbf{x}}

\newcommand{\bm}{m}

\newcommand{\dmp}{k}

\newcommand{\bC}{\mathbb{C}}
\newcommand{\bR}{\mathbb{R}}
\newcommand{\bK}{\mathbb{K}}

\newcommand{\sz}{n}

\title[Set-theoretical solutions to the Zamolodchikov tetrahedron equation]%
{Set-theoretical solutions to the Zamolodchikov tetrahedron equation on associative rings \\
and Liouville integrability}
\date{}

\author{Sergei Igonin} 
\address{Centre of Integrable Systems, P.G. Demidov Yaroslavl State University, Yaroslavl, Russia}

\author{}
\email{s-igonin@yandex.ru}

\begin{document}

\begin{abstract}

This paper is devoted to tetrahedron maps, which are set-theoretical solutions 
to the Zamolodchikov tetrahedron equation.
We construct a family of tetrahedron maps on associative rings.

We show that matrix tetrahedron maps presented in~[arXiv:2110.05998]
are a particular case of our construction.
This provides an algebraic explanation of the fact that 
the matrix maps from~[arXiv:2110.05998] satisfy the tetrahedron equation.

Also, Liouville integrability is established for some of the constructed maps.

\end{abstract}

\maketitle

\bigskip

\hspace{-.6cm} \textbf{Mathematics Subject Classification 2020:} 16T25, 81R12.


\hspace{-.6cm} \textbf{Keywords:} 
Zamolodchikov tetrahedron equation, tetrahedron maps, associative rings, \mbox{Liouville} integrability

\section{Introduction}

\label{sintr}

This paper is devoted to tetrahedron maps, which are set-theoretical solutions 
to the Zamolodchikov tetrahedron equation~\cite{Zamolodchikov,Zamolodchikov-2}.
The tetrahedron equation belongs to fundamental equations of mathematical physics 
and has applications in many diverse branches of physics and mathematics, 
including statistical mechanics, quantum field theories, combinatorics, 
low-dimensional topology, and the theory of integrable systems
(see, e.g.,~\cite{Bazhanov-Sergeev,Bazhanov-Sergeev-2,Kashaev-Sergeev,Kassotakis-Tetrahedron,Nijhoff,TalalUMN21} and references therein).

Presently, the relations of tetrahedron maps with integrable systems 
and with algebraic structures (including groups and rings) are a very active area of research 
(see, e.g.,~\cite{Doliwa-Kashaev,igkr21,Kassotakis-Tetrahedron,Sharygin-Talalaev,TalalUMN21,Yoneyama21}).

In this paper we construct tetrahedron maps 
on associative rings and study their properties.
In particular, for any associative ring~$\mathcal{A}$ 
and any element $M\in\mathcal{A}$ we define the maps~\er{xmz},~\er{zmx}
and prove that they satisfy the tetrahedron equation~\er{tetr-eq}.
The obtained tetrahedron maps~\er{xmz},~\er{zmx} are new, to our knowledge.

Also, 
we study the matrix tetrahedron maps~\er{mc1},~\er{mc2} constructed in~\cite{igkr21}.
The preprint~\cite{igkr21} says that one can verify by a straightforward computation 
that \eqref{mc1},~\eqref{mc2} satisfy the tetrahedron equation~\er{tetr-eq}.
In Theorem~\ref{talex} we present an algebraic explanation of this fact.

Namely, in Theorem~\ref{talex} we show that the maps~\eqref{mc1},~\eqref{mc2} 
are of the form~\er{xmz} for $\mathcal{A}=\mat_{\sz}(\bK)$ 
and some matrices $M\in\mat_{\sz}(\bK)$.
Here $\bK$ is a field, and $\mat_{\sz}(\bK)$ is 
the associative ring of $\sz\times\sz$ matrices with entries from~$\bK$.
Usually, $\bK$ is either $\bR$ or $\bC$, but one can consider also arbitrary fields.

Proposition~\ref{pp13} recalls a well-known construction, 
which allows one to obtain a new tetrahedron map from a known one.
Applying Proposition~\ref{pp13} to the maps~\er{mc1},~\er{mc2}, 
we get new tetrahedron maps~\er{mc1zx},~\er{mc2zx}.

In Section~\ref{slint} we prove Liouville integrability
\begin{itemize}
	\item for the map~\eqref{mc1zx} in the case $1\le\bm\le\sz/2$ in Proposition~\ref{tlimc1zx},
	\item for the map~\eqref{mc2zx} in the case $3\sz/2\le\hat{\bm}\le 2\sz-1$ in Proposition~\ref{tlimc2zx}.
\end{itemize}
The proofs of these results on Liouville integrability for~\er{mc1zx},~\er{mc2zx} 
are deduced from the proofs of similar results for~\er{mc1},~\er{mc2} taken from~\cite{igkr21}.

We use the standard notion of Liouville integrability for maps on manifolds
(see, e.g.,~\cite{fordy14,igkr21,Sokor-Sasha,vesel1991} and references therein).
This notion is presented in Definition~\ref{dli}.

\section{Tetrahedron maps on sets and associative rings}
\lb{stmar}

Let $\md$ be a set. A \emph{tetrahedron map} is a map 
\begin{equation}
\notag
T\cl \md^3\rightarrow \md^3,\qquad
T(x,y,z)=\big(f(x,y,z),g(x,y,z),h(x,y,z)\big),\qquad x,y,z\in \md,
\end{equation}
satisfying the (Zamolodchikov) \emph{tetrahedron equation}
\begin{equation}\label{tetr-eq}
    T^{123}\circ T^{145} \circ T^{246}\circ T^{356}=T^{356}\circ T^{246}\circ T^{145}\circ T^{123}.
\end{equation}
Here $T^{ijk}\cl \md^6\rightarrow \md^6$ for $i,j,k=1,\ldots,6$, $i<j<k$, is the map 
acting as $T$ on the $i$th, $j$th, $k$th factors 
of the Cartesian product $\md^6$ and acting as identity on the remaining factors.
For instance,
$$
T^{246}(x,y,z,r,s,t)=\big(x,f(y,r,t),z,g(y,r,t),s,h(y,r,t)\big),\qqquad x,y,z,r,s,t\in \md.
$$

The statement of Proposition~\ref{pp13} is well known.
A proof can be found, e.g., in~\cite{Kassotakis-Tetrahedron}.
\begin{proposition}
\label{pp13}
Consider the permutation map
\begin{gather*}
P^{13}\cl \md^3\to \md^3,\qqquad
P^{13}(a_1,a_2,a_3)=(a_3,a_2,a_1),\qquad a_i\in \md.
\end{gather*}
If a map $T\cl \md^3\to \md^3$ satisfies the tetrahedron equation~\eqref{tetr-eq}
then $\tilde{T}=P^{13}\circ T\circ P^{13}$ obeys this equation as well.
\end{proposition}

\begin{proposition}
\label{prnu}
For any associative ring $\mathcal{A}$, we have the tetrahedron maps
\begin{gather}
\label{namnK}
\tbf\cl(\mathcal{A})^3\to(\mathcal{A})^3,\qqquad
\tbf(X,Y,Z)=(X,\,Y+XZ,\,Z),\\
\label{ntmnK}
\tilde{\tbf}\cl(\mathcal{A})^3\to(\mathcal{A})^3,\qqquad
\tilde{\tbf}(X,Y,Z)=(X,\,Y+ZX,\,Z),\\
\notag
X,Y,Z\in \mathcal{A}.
\end{gather}
\end{proposition}
\begin{proof} For the map~\er{namnK} and 
any elements $A_1,A_2,A_3,A_4,A_5,A_6\in\mathcal{A}$ we have
\begin{multline*}
(\tbf^{123}\circ \tbf^{145} \circ \tbf^{246}\circ \tbf^{356})(A_1,A_2,A_3,A_4,A_5,A_6)=\\
=(\tbf^{123}\circ \tbf^{145} \circ \tbf^{246})(A_1,A_2,A_3,A_4,A_5+A_3A_6,A_6)=\\
=(\tbf^{123}\circ \tbf^{145})(A_1,A_2,A_3,A_4+A_2A_6,A_5+A_3A_6,A_6)=\\
=\tbf^{123}(A_1,A_2,A_3,A_4+A_2A_6+A_1(A_5+A_3A_6),A_5+A_3A_6,A_6)=\\
=(A_1,A_2+A_1A_3,A_3,A_4+A_2A_6+A_1A_5+A_1A_3A_6,A_5+A_3A_6,A_6),
\end{multline*}
\begin{multline*}
(\tbf^{356}\circ \tbf^{246}\circ \tbf^{145}\circ \tbf^{123})(A_1,A_2,A_3,A_4,A_5,A_6)=\\
=(\tbf^{356}\circ \tbf^{246}\circ \tbf^{145})(A_1,A_2+A_1A_3,A_3,A_4,A_5,A_6)=\\
=(\tbf^{356}\circ \tbf^{246})(A_1,A_2+A_1A_3,A_3,A_4+A_1A_5,A_5,A_6)=\\
=\tbf^{356}(A_1,A_2+A_1A_3,A_3,A_4+A_1A_5+(A_2+A_1A_3)A_6,A_5,A_6)=\\
=(A_1,A_2+A_1A_3,A_3,A_4+A_1A_5+A_2A_6+A_1A_3A_6,A_5+A_3A_6,A_6).
\end{multline*}
Therefore, the map~\er{namnK} satisfies the tetrahedron equation~\eqref{tetr-eq}.

The map~\er{ntmnK} is obtained from~\er{namnK} by means of Proposition~\ref{pp13}.
\end{proof}

\begin{theorem}
\label{tam}
Let $\mathcal{A}$ be an associative ring. Let $M\in\mathcal{A}$. 
The maps
\begin{gather}
\label{xmz}
\tbf_M\cl(\mathcal{A})^3\to(\mathcal{A})^3,\qquad
\tbf_M(X,Y,Z)=(X,Y+XMZ,Z),\\
\label{zmx}
\tilde{\tbf}_M\cl(\mathcal{A})^3\to(\mathcal{A})^3,\qquad
\tilde{\tbf}_M(X,Y,Z)=(X,Y+ZMX,Z),\\
\notag
X,Y,Z\in\mathcal{A},
\end{gather}
satisfy the tetrahedron equation.
\end{theorem}
\begin{proof}
On the associative ring~$\mathcal{A}$ one has the addition and multiplication operations, 
which we call the initial addition and multiplication on~$\mathcal{A}$.
Let us consider a new multiplication operation on~$\mathcal{A}$
defined as follows 
\begin{gather}
\label{newm}
P*Q:=PMQ,\qqquad P,Q\in\mathcal{A}.
\end{gather}
Since the initial multiplication on~$\mathcal{A}$ is associative, 
the new multiplication~\er{newm} is associative as well.

Moreover, the set~$\mathcal{A}$ with the initial addition and 
the new multiplication~\er{newm} is an associative ring.
Applying Proposition~\ref{prnu} to this ring, we get the tetrahedron maps
\begin{gather}
\label{xzstar}
\tbf\cl(\mathcal{A})^3\to(\mathcal{A})^3,\qqquad
\tbf(X,Y,Z)=(X,\,Y+X*Z,\,Z),\\
\label{zxstar}
\tilde{\tbf}\cl(\mathcal{A})^3\to(\mathcal{A})^3,\qqquad
\tilde{\tbf}(X,Y,Z)=(X,\,Y+Z*X,\,Z),\\
\notag
X,Y,Z\in \mathcal{A}.
\end{gather}
Taking into account formula~\er{newm}, 
we see that the maps~\er{xmz},~\er{zmx} coincide with the tetrahedron maps~\er{xzstar},~\er{zxstar}.
\end{proof}

We denote by~$\zsp$ the set of positive integers. 
\begin{corollary}
\label{ctm}
Let $\sz\in\zsp$ and $M\in\mat_{\sz}(\bK)$. The maps 
\begin{gather}
\label{yxmz}
\tbf_M\cl\big(\mat_{\sz}(\bK)\big)^3\to\big(\mat_{\sz}(\bK)\big)^3,\qquad
\tbf_M(X,Y,Z)=(X,Y+XMZ,Z),\\
\label{yzmx}
\tilde{\tbf}_M\cl\big(\mat_{\sz}(\bK)\big)^3\to\big(\mat_{\sz}(\bK)\big)^3,\qquad
\tilde{\tbf}_M(X,Y,Z)=(X,Y+ZMX,Z),\\
\notag
X,Y,Z\in\mat_{\sz}(\bK),
\end{gather}
satisfy the tetrahedron equation.
\end{corollary}
\begin{proof}
This follows from Theorem~\ref{tam} in the case $\mathcal{A}=\mat_{\sz}(\bK)$.
\end{proof}
\begin{remark}
\lb{rtmk}
Let $\hat{\mathcal{A}}$ be an associative ring with a unit~$e\in\hat{\mathcal{A}}$.
(That is, for any $x\in\hat{\mathcal{A}}$ we have $xe=ex=x$.)
Let $K\in\hat{\mathcal{A}}$ be an invertible element.
(That is, there is $K^{-1}\in\hat{\mathcal{A}}$ such that $KK^{-1}=K^{-1}K=e$.)
The preprint~\cite{igkr21} presents the following tetrahedron maps
\begin{gather}
\label{amnK}
\tm_K\cl\big(\hat{\mathcal{A}}\big)^3\to\big(\hat{\mathcal{A}}\big)^3,\qqquad
\tm_K(X,Y,Z)=(X,\,YK+XZ,\,KZK^{-1}),\\
\label{atmnK}
\tilde{\tm}_K\cl\big(\hat{\mathcal{A}}\big)^3\to\big(\hat{\mathcal{A}}\big)^3,\qqquad
\tilde{\tm}_K(X,Y,Z)=(KXK^{-1},\,YK+ZX,\,Z),\\
\notag
X,Y,Z\in\hat{\mathcal{A}}.
\end{gather}

In the case $K=e$ formulas~\er{amnK},~\er{atmnK} coincide with~\er{namnK},~\er{ntmnK}.
Note that the maps~\er{namnK},~\er{ntmnK} are defined for any associative ring~$\mathcal{A}$,
which does not necessarily have a unit.

As discussed in~\cite{igkr21}, 
in the case $\hat{\mathcal{A}}=\bC$ the maps~\er{amnK},~\er{atmnK} 
reduce to a tetrahedron map from Sergeev's classification~\cite{Sergeev}.
\end{remark}

Let
\begin{gather}
\notag
\sz\in\zsp,\qqquad\bm\in\{1,\dots,\sz\},
\qqquad\hat{\bm}\in\{\sz+1,\dots,2\sz-1\}.
\end{gather}
Below the elements of matrices $X,Z\in\mat_{\sz}(\bK)$ are denoted by $x_{k,l}$, $z_{k,l}$ 
for $k,l=1,\dots,\sz$, and we use the multiplication of a column by a row
\begin{gather}
\lb{xez}
\begin{pmatrix}
x_{1,i}\\
x_{2,i}\\
\vdots\\
x_{\sz,i}
\end{pmatrix}(z_{j,1}\ z_{j,2}\ \dots\ z_{j,\sz})=
\begin{pmatrix}
    x_{1,i}z_{j,1} & x_{1,i}z_{j,2}&\hdots  &x_{1,i}z_{j,n} \\
		x_{2,i}z_{j,1} & x_{2,i}z_{j,2}&\hdots  & x_{2,i}z_{j,n} \\
		\vdots & \vdots & \ddots &\vdots\\
		x_{n,i}z_{j,1} & x_{n,i}z_{j,2}&\hdots  & x_{n,i}z_{j,n}
\end{pmatrix}=XE_{i,j}Z,\quad i,j\in\{1,\dots,\sz\}.
\end{gather}
Here $E_{i,j}\in\mat_{\sz}(\bK)$ is the $\sz\times\sz$ matrix with 
$(i,j)$-th entry equal to~1 and all other entries equal to zero.

The preprint~\cite{igkr21} presents the following maps
\begin{gather}
\label{mc1}
\begin{gathered}
\ftm_{\sz,\bm}\cl\big(\mat_{\sz}(\bK)\big)^3\to\big(\mat_{\sz}(\bK)\big)^3,
\qqquad\bm\in\{1,\dots,\sz\},\\
\ftm_{\sz,\bm}(X,Y,Z)=\left(X,\,Y+\sum_{i=1}^\bm\begin{pmatrix}
x_{1,i}\\
x_{2,i}\\
\vdots\\
x_{\sz,i}
\end{pmatrix}(z_{\sz-\bm+i,1}\ z_{\sz-\bm+i,2}\ \dots\ z_{\sz-\bm+i,\sz}),\,Z\right),\\
\end{gathered}\\
\label{mc2}
\begin{gathered}
\stm_{\sz,\hat{\bm}}\cl\big(\mat_{\sz}(\bK)\big)^3\to\big(\mat_{\sz}(\bK)\big)^3,
\qqquad\hat{\bm}\in\{\sz+1,\dots,2\sz-1\},\\
\stm_{\sz,\hat{\bm}}(X,Y,Z)=
\left(X,\,Y+\sum_{i=1}^{2\sz-\hat{\bm}}\begin{pmatrix}
x_{1,\hat{\bm}-\sz+i}\\
x_{2,\hat{\bm}-\sz+i}\\
\vdots\\
x_{\sz,\hat{\bm}-\sz+i}
\end{pmatrix}(z_{i,1}\ z_{i,2}\ \dots\ z_{i,\sz}),\,Z\right),
\end{gathered}\\
\notag
X,Y,Z\in\mat_{\sz}(\bK).
\end{gather}
The preprint~\cite{igkr21} says that one can verify by a straightforward computation 
that the maps~\eqref{mc1},~\eqref{mc2} satisfy the tetrahedron equation.
In Theorem~\ref{talex} below we present an algebraic explanation of this fact.

\begin{theorem}
\label{talex}
The map~\eqref{mc1} is of the form~\er{yxmz} for $M=\sum_{i=1}^\bm E_{i,\sz-\bm+i}$.

The map~\eqref{mc2} is of the form~\er{yxmz} for $M=\sum_{i=1}^{2\sz-\hat{\bm}}E_{\hat{\bm}-\sz+i,i}$.

Therefore, Theorem~\textup{\ref{tam}} and Corollary~\textup{\ref{ctm}} 
explain why~\eqref{mc1},~\eqref{mc2} obey the tetrahedron equation.
\end{theorem}
\begin{proof}
The statements follow from formula~\er{xez}.
\end{proof}

\begin{proposition}
\label{pttm}
Let $\sz\in\zsp$, $\,\bm\in\{1,\dots,\sz\}$, $\,\hat{\bm}\in\{\sz+1,\dots,2\sz-1\}$.
One has the following tetrahedron maps
\begin{gather}
\label{mc1zx}
\begin{gathered}
\tilde{\ftm}_{\sz,\bm}\cl\big(\mat_{\sz}(\bK)\big)^3\to\big(\mat_{\sz}(\bK)\big)^3,
\qqquad\bm\in\{1,\dots,\sz\},\\
\tilde{\ftm}_{\sz,\bm}(X,Y,Z)=\left(X,\,Y+\sum_{i=1}^\bm\begin{pmatrix}
z_{1,i}\\
z_{2,i}\\
\vdots\\
z_{\sz,i}
\end{pmatrix}(x_{\sz-\bm+i,1}\ x_{\sz-\bm+i,2}\ \dots\ x_{\sz-\bm+i,\sz}),\,Z\right),\\
\end{gathered}\\
\label{mc2zx}
\begin{gathered}
\tilde{\stm}_{\sz,\hat{\bm}}\cl\big(\mat_{\sz}(\bK)\big)^3\to\big(\mat_{\sz}(\bK)\big)^3,
\qqquad\hat{\bm}\in\{\sz+1,\dots,2\sz-1\},\\
\tilde{\stm}_{\sz,\hat{\bm}}(X,Y,Z)=
\left(X,\,Y+\sum_{i=1}^{2\sz-\bm}\begin{pmatrix}
z_{1,\bm-\sz+i}\\
z_{2,\bm-\sz+i}\\
\vdots\\
z_{\sz,\bm-\sz+i}
\end{pmatrix}(x_{i,1}\ x_{i,2}\ \dots\ x_{i,\sz}),\,Z\right),
\end{gathered}\\
\notag
X,Y,Z\in\mat_{\sz}(\bK).
\end{gather}
\end{proposition}
\begin{proof} 
The maps~\er{mc1zx},~\er{mc2zx} are obtained from~\er{mc1},~\er{mc2} by means of Proposition~\ref{pp13}.
\end{proof}

\section{Liouville integrability}
\label{slint} 

In Definition~\ref{dli} we recall the standard notion of Liouville integrability for maps on manifolds
(see, e.g.,~\cite{fordy14,igkr21,Sokor-Sasha,vesel1991} and references therein).
\begin{definition}
\label{dli}
Let $\dmp\in\zsp$. 
Let $\Mfd$ be a $\dmp$-dimensional manifold 
with (local) coordinates $\xc_1,\dots,\xc_\dmp$.
A (smooth or analytic) map $F\cl \Mfd\to \Mfd$ is said to be \emph{Liouville integrable} 
(or \emph{completely integrable}) if 
one has the following objects on the manifold~$\Mfd$.
\begin{itemize}
	\item A Poisson bracket $\{\,,\,\}$ which is 
	invariant under the map~$F$ and is of constant rank~$2r$ 
	for some positive integer~$r\le\dmp/2$ (i.e., the $\dmp\times\dmp$ matrix with the entries 
	$\{\xc_i,\xc_j\}$ is of constant rank~$2r$).
	The invariance of the bracket means the following.
	For any functions $g$, $h$ on~$\Mfd$ one has 
\begin{gather}
\label{ipbr}
	\{g,h\}\circ F=\{g\circ F,\,h\circ F\}.
\end{gather}
To prove that the bracket is invariant,
it is sufficient to check property~\eqref{ipbr} for $g=\xc_i$, $\,h=\xc_j$, $\,i,j=1,\dots,\dmp$.
	
	In our examples considered below the manifold has a system of coordinates 
	$\xc_1,\dots,\xc_{\dmp}$ such that for any $i,j=1,\dots,\dmp$ 
	the function $\{\xc_i,\xc_j\}$ is constant. 
	Then, in order to prove that the bracket is invariant under~$F$, 
	it is sufficient to show that $\{\xc_i\circ F,\,\xc_j\circ F\}=\{\xc_i,\xc_j\}$ for all $i$, $j$.
	
\item If $2r<\dmp$ then one needs also $\dmp-2r$ functions 
	\begin{gather}
\label{cfs}
C_s,\qqquad s=1,\dots,\dmp-2r,
\end{gather}
	which are invariant under~$F$ (i.e., $C_s\circ F=C_s$) 
	and are Casimir functions (i.e., $\{C_s,g\}=0$ for any function~$g$).
	\item One has $r$ functions 
	\begin{gather}
\label{finv}
I_l,\qqquad l=1,\dots,r,
\end{gather}
	which are invariant under~$F$
	and are in involution with respect to the Poisson bracket (i.e., $\{I_{l_1},I_{l_2}\}=0$ 
	for all $l_1,l_2=1,\dots,r$).
	\item The functions~\eqref{cfs},~\eqref{finv} must be functionally independent.
\end{itemize}	
\end{definition}
	
In this section we assume that $\bK$ is either $\bR$ or $\bC$.
Then the set $\big(\mat_{\sz}(\bK)\big)^3$ is a manifold.
The elements of matrices $X,Y,Z\in\big(\mat_{\sz}(\bK)\big)^3$ 
are denoted by $x_{i,j}$, $y_{i,j}$, $z_{i,j}$ for $i,j=1,\dots,\sz$.
Then 
\begin{gather}
\label{lcs}
x_{i,j},\quad y_{i,j},\quad z_{i,j},\qqquad
i,j=1,\dots,\sz, 
\end{gather}
can be regarded as coordinates 
on the $3\sz^2$-dimensional manifold $\big(\mat_{\sz}(\bK)\big)^3$.
Clearly, the functions $x_{i,j}$, $z_{i,j}$ are invariant under 
the maps \eqref{mc1}, \eqref{mc2}, \eqref{mc1zx}, \eqref{mc2zx}.

Propositions~\ref{emnmc1},~\ref{emnmc2} below are proved in the preprint~\cite{igkr21}.
To clarify the main ideas, we present the proofs from~\cite{igkr21}.

\begin{proposition}[\cite{igkr21}]
\label{emnmc1}
Let $\sz,\bm\in\zsp$ such that $1\le\bm\le\sz/2$.
Then the map~\eqref{mc1} is Liouville integrable.
\end{proposition}
\begin{proof}
On the manifold $\big(\mat_{\sz}(\bK)\big)^3$ we consider the Poisson bracket defined as follows.
The bracket of two coordinates from the list~\eqref{lcs} is nonzero only for
\begin{gather}
\label{pbmc1}
\begin{gathered}
\{y_{p,j},z_{p,j}\}=-\{z_{p,j},y_{p,j}\}=1,\qqquad p=1,\dots,\sz-\bm,\qquad
j=1,\dots,\sz,\\
\{y_{\sz-\bm+q,j},x_{j,\bm+q}\}=-\{x_{j,\bm+q},y_{\sz-\bm+q,j}\}=1,\qqquad q=1,\dots,\bm,\qquad
j=1,\dots,\sz.
\end{gathered}
\end{gather}
This Poisson bracket is of rank~$2\sz^2$.
The $\sz^2$ functions
\begin{gather}
\label{cfmc1}
z_{\sz-\bm+q,j},\quad x_{j,q},\quad x_{j,2\bm+s},\qquad
q=1,\dots,\bm,\quad s=1,\dots,\sz-2\bm,\quad j=1,\dots,\sz,
\end{gather}
are Casimir functions, since they do not appear in~\eqref{pbmc1}.
The rank of the bracket plus 
the number of the Casimir functions~\eqref{cfmc1} equals the dimension of the manifold. 

The $\sz^2$ functions
\begin{gather}
\label{fimc1}
z_{p,j},\qquad x_{j,\bm+q},\qqquad 
p=1,\dots,\sz-\bm,\qquad q=1,\dots,\bm,\qquad j=1,\dots,\sz,
\end{gather}
are in involution with respect to the Poisson bracket, 
as one has $\{x_{i,j},x_{i',j'}\}=\{z_{i,j},z_{i',j'}\}=\{x_{i,j},z_{i',j'}\}=0$ for all $i,j,i',j'=1,\dots,\sz$.
The functions~\eqref{cfmc1},~\eqref{fimc1} are functionally independent and 
are invariant under the map~\eqref{mc1}.

Let us show that the bracket~\eqref{pbmc1} is invariant as well.
In~\eqref{mc1} we have
\begin{gather}
\label{cmc1}
\begin{gathered}
\ftm_{\sz,\bm}(X,Y,Z)=(\hat{X},\hat{Y},\hat{Z}),\\
\hat{X}=X,\qqquad
\hat{Y}=Y+\sum_{i=1}^\bm\begin{pmatrix}
x_{1,i}\\
x_{2,i}\\
\vdots\\
x_{\sz,i}
\end{pmatrix}(z_{\sz-\bm+i,1}\ z_{\sz-\bm+i,2}\ \dots\ z_{\sz-\bm+i,\sz}),\qqquad
\hat{Z}=Z.
\end{gathered}
\end{gather}
The elements of matrices $\hat{X},\hat{Y},\hat{Z}\in\big(\mat_{\sz}(\bK)\big)^3$ 
are denoted by $\hat{x}_{i,j}$, $\hat{y}_{i,j}$, $\hat{z}_{i,j}$ for $i,j=1,\dots,\sz$.
The functions 
\begin{gather}
\label{xizi}
x_{1,i},\ x_{2,i},\ \dots,\ x_{\sz,i},\quad
z_{\sz-\bm+i,1},\ z_{\sz-\bm+i,2},\ \dots,\ z_{\sz-\bm+i,\sz},\qqquad i=1,\dots,\bm,
\end{gather}
which appear in~\eqref{cmc1}, belong to the list of the Casimir functions~\eqref{cfmc1}.
Therefore, formulas~\eqref{pbmc1},~\eqref{cmc1} yield
\begin{gather*}
\{\hat{x}_{i,j},\hat{x}_{i',j'}\}=\{x_{i,j},x_{i',j'}\},\quad
\{\hat{y}_{i,j},\hat{y}_{i',j'}\}=\{y_{i,j},y_{i',j'}\},\quad
\{\hat{z}_{i,j},\hat{z}_{i',j'}\}=\{z_{i,j},z_{i',j'}\},\\
\{\hat{x}_{i,j},\hat{y}_{i',j'}\}=\{x_{i,j},y_{i',j'}\},\quad
\{\hat{x}_{i,j},\hat{z}_{i',j'}\}=\{x_{i,j},z_{i',j'}\},\quad
\{\hat{y}_{i,j},\hat{z}_{i',j'}\}=\{y_{i,j},z_{i',j'}\},\\
i,j,i',j'=1,\dots,\sz,
\end{gather*}
which implies that the Poisson bracket is invariant under the map~\eqref{mc1}.

Therefore, the map~\eqref{mc1} in the case $1\le\bm\le\sz/2$ is Liouville integrable.
\end{proof}

\begin{example}
\label{en2m1}
Let $\sz=2$ and $\bm=1$. Then the map~\eqref{mc1} reads
\begin{gather*}
\begin{gathered}
\ftm_{2,1}\cl\big(\mat_{2}(\bK)\big)^3\to\big(\mat_{2}(\bK)\big)^3,\\
\ftm_{2,1}(X,Y,Z)=
\left(X,\,Y+\begin{pmatrix}
    x_{11}z_{21} & x_{11}z_{22}\\
		x_{21}z_{21} & x_{21}z_{22}
\end{pmatrix},\,Z\right),\\
X=\begin{pmatrix}
    x_{11} & x_{12}\\
		x_{21} & x_{22}
\end{pmatrix},\qqquad 
Y=\begin{pmatrix}
    y_{11} & y_{12}\\
		y_{21} & y_{22}
\end{pmatrix},\qqquad
Z=\begin{pmatrix}
    z_{11} & z_{12}\\
		z_{21} & z_{22}
\end{pmatrix}.
\end{gathered}
\end{gather*}
One has the coordinates 
$x_{ij},y_{ij},z_{ij}$, $i,j\in\{1,2\}$, on the $12$-dimensional 
manifold $\big(\mat_{2}(\bK)\big)^3$.

The Poisson bracket~\eqref{pbmc1} reads
\begin{gather}
\label{pbsc22}
\begin{gathered}
\{y_{11},z_{11}\}=-\{z_{11},y_{11}\}=1,\qqquad
\{y_{12},z_{12}\}=-\{z_{12},y_{12}\}=1,\\
\{y_{21},x_{12}\}=-\{x_{12},y_{21}\}=1,\qqquad
\{y_{22},x_{12}\}=-\{x_{12},y_{22}\}=1,
\end{gathered}
\end{gather}
and is of rank~$8$.
The Casimir functions~\eqref{cfmc1} are $z_{21}$, $z_{22}$, $x_{11}$, $x_{21}$.
The functions~\eqref{fimc1} are $z_{11}$, $z_{12}$, $x_{12}$, $x_{22}$,
and they are in involution with respect to the bracket~\eqref{pbsc22}.
\end{example}

\begin{example}
\label{en5m2}
Now let $\sz=5$ and $\bm=2$. Then the map~\eqref{mc1} is
\begin{gather*}
\begin{gathered}
\ftm_{5,2}\cl\big(\mat_{5}(\bK)\big)^3\to\big(\mat_{5}(\bK)\big)^3,\\
\ftm_{5,2}(X,Y,Z)=
\left(X,\,Y+
\begin{pmatrix}
          x_{11}z_{41}+x_{12}z_{51}& \hdots&x_{11}z_{45}+x_{12}z_{55}\\
          \vdots&\ddots&\vdots\\
          x_{51}z_{41}+x_{52}z_{51}&\hdots&x_{51}z_{45}+x_{52}z_{55}
\end{pmatrix},\,Z\right),\\
X=\begin{pmatrix}
          x_{11}& \hdots&x_{15}\\
          \vdots&\ddots&\vdots\\
          x_{51}&\hdots&x_{55}
\end{pmatrix},\qqquad 
Y=\begin{pmatrix}
          y_{11}& \hdots&y_{15}\\
          \vdots&\ddots&\vdots\\
          y_{51}&\hdots&y_{55}
\end{pmatrix},\qqquad
Z=\begin{pmatrix}
          z_{11}& \hdots&z_{15}\\
          \vdots&\ddots&\vdots\\
          z_{51}&\hdots&z_{55}
\end{pmatrix}.
\end{gathered}
\end{gather*}
One has the coordinates 
$x_{ij},y_{ij},z_{ij}$, $i,j\in\{1,\dots,5\}$, on the $75$-dimensional 
manifold $\big(\mat_{5}(\bK)\big)^3$.

The Poisson bracket~\eqref{pbmc1} reads
\begin{gather}
\label{pb52}
\begin{gathered}
\{y_{1,j},z_{1,j}\}=-\{z_{1,j},y_{1,j}\}=
\{y_{2,j},z_{2,j}\}=-\{z_{2,j},y_{2,j}\}=
\{y_{3,j},z_{3,j}\}=-\{z_{3,j},y_{3,j}\}=1,\\
\{y_{4,j},x_{j,3}\}=-\{x_{j,3},y_{4,j}\}=
\{y_{5,j},x_{j,4}\}=-\{x_{j,4},y_{5,j}\}=1,\qquad j=1,\dots,5,
\end{gathered}
\end{gather}
and is of rank~$50$.
The Casimir functions~\eqref{cfmc1} are $z_{4,j}$, $z_{5,j}$,
$x_{j,1}$, $x_{j,2}$, $x_{j,5}$, $j=1,\dots,5$.
The functions~\eqref{fimc1} are $z_{1,j}$, $z_{2,j}$, $z_{3,j}$, $x_{j,3}$, $x_{j,4}$,
and they are in involution with respect to the bracket~\eqref{pb52}.
\end{example}

\begin{proposition}[\cite{igkr21}]
\label{emnmc2}
Let $\sz,\hat{\bm}\in\zsp$ such that $3\sz/2\le\hat{\bm}\le 2\sz-1$.
Then the map~\eqref{mc2} is Liouville integrable.
\end{proposition}
\begin{proof}
Now on the manifold $\big(\mat_{\sz}(\bK)\big)^3$ we consider the following Poisson bracket.
The bracket of two coordinates from the list~\eqref{lcs} is nonzero only for
\begin{gather}
\label{pbmc2}
\begin{gathered}
\{y_{p,j},z_{2\sz-\hat{\bm}+p,j}\}=-\{z_{2\sz-\hat{\bm}+p,j},y_{p,j}\}=1,\qqquad p=1,\dots,\hat{\bm}-\sz,\qquad
j=1,\dots,\sz,\\
\{y_{\hat{\bm}-\sz+q,j},x_{j,q}\}=-\{x_{j,q},y_{\sz-\hat{\bm}+q,j}\}=1,\qqquad q=1,\dots,2\sz-\hat{\bm},\qquad
j=1,\dots,\sz.
\end{gathered}
\end{gather}
This bracket is of rank~$2\sz^2$.
The $\sz^2$ functions
\begin{gather}
\label{cfmc2}
z_{q,j},\quad x_{j,2\sz-\hat{\bm}+p},\qqquad
q=1,\dots,2\sz-\hat{\bm},\quad p=1,\dots,\hat{\bm}-\sz,\quad j=1,\dots,\sz,
\end{gather}
are Casimir functions, as they do not appear in~\eqref{pbmc2}.
The rank of the bracket plus 
the number of the Casimir functions~\eqref{cfmc2} equals the dimension of the manifold. 

The $\sz^2$ functions
\begin{gather}
\label{fimc2}
z_{2\sz-\hat{\bm}+p,j},\quad x_{j,q},\qqquad
p=1,\dots,\hat{\bm}-\sz,\qquad
q=1,\dots,2\sz-\hat{\bm},\qquad j=1,\dots,\sz,
\end{gather}
are in involution with respect to the bracket, 
since we have $\{x_{i,j},x_{i',j'}\}=\{z_{i,j},z_{i',j'}\}=\{x_{i,j},z_{i',j'}\}=0$ for all $i,j,i',j'=1,\dots,\sz$.
The functions~\eqref{cfmc2},~\eqref{fimc2} are functionally independent and 
are invariant under the map~\eqref{mc2}.

Similarly to the proof of Proposition~\ref{emnmc1}, 
one shows that the bracket~\eqref{pbmc2} is invariant as well.
Hence the map~\eqref{mc2} in the case $3\sz/2\le\hat{\bm}\le 2\sz-1$ is Liouville integrable.
\end{proof}

\begin{proposition}
\label{tlimc1zx}
Let $\sz,\bm\in\zsp$ such that $1\le\bm\le\sz/2$.
Then the map~\eqref{mc1zx} is Liouville integrable.
\end{proposition}
\begin{proof}
To obtain a proof for Proposition~\ref{tlimc1zx}, 
it is sufficient to interchange~$x_{i,j}$ with~$z_{i,j}$ in the proof of Proposition~\ref{emnmc1}.
\end{proof}

\begin{proposition}
\label{tlimc2zx}
Let $\sz,\hat{\bm}\in\zsp$ such that $3\sz/2\le\hat{\bm}\le 2\sz-1$.
Then the map~\eqref{mc2zx} is Liouville integrable.
\end{proposition}
\begin{proof}
To obtain a proof for Proposition~\ref{tlimc2zx}, 
it is sufficient to interchange~$x_{i,j}$ with~$z_{i,j}$ in the proof of Proposition~\ref{emnmc2}.
\end{proof}

\section*{Acknowledgements} 

The work on Sections~\ref{sintr},~\ref{stmar}
was supported by the Russian Science Foundation (grant No. 21-71-30011).

The work on Section~\ref{slint}
was carried out within the framework of a development programme for the Regional Scientific and Educational Mathematical Centre of the P.G. Demidov Yaroslavl State University with financial support from the Ministry of Science and Higher Education 
of the Russian Federation (Agreement on provision of subsidy from the federal budget No. 075-02-2022-886).

The author would like to thank S.~Konstantinou-Rizos, S.M.~Sergeev, and D.V.~Talalaev for useful discussions.


\begin{thebibliography}{100}

\bibitem{Bazhanov-Sergeev}
V.V. Bazhanov and S.M. Sergeev.
Zamolodchikov's tetrahedron equation and hidden structure of quantum groups.
\emph{J. Phys. A} \textbf{39} (2006), 3295--3310.

\bibitem{Bazhanov-Sergeev-2}
V.V. Bazhanov, V.V. Mangazeev, and S.M. Sergeev.
Quantum geometry of three-dimensional lattices. \emph{J. Stat. Mech.} (2008) {P07004}.

\bibitem{Doliwa-Kashaev}
A. Doliwa and R.M. Kashaev.
Non-commutative birational maps satisfying Zamolodchikov equation, and Desargues lattices.
\emph{J. Math. Phys.} \textbf{61} (2020) 092704.

\bibitem{fordy14}
A.P.~Fordy. Periodic cluster mutations and related integrable maps.
\emph{J. Phys. A: Math. Theor.} \textbf{47} (2014) 474003.

\bibitem{igkr21}
S.~Igonin and S.~Konstantinou-Rizos.
Algebraic and differential-geometric constructions of set-theoretical solutions 
to the Zamolodchikov tetrahedron equation.
arXiv:2110.05998 (2021).

\bibitem{Kashaev-Sergeev}
R.M. Kashaev, I.G. Koperanov, and S.M. Sergeev. 
Functional tetrahedron equation.
\emph{Theor. Math. Phys.} \textbf{117} (1998), 1402--1413.

\bibitem{Kassotakis-Tetrahedron}
P.~Kassotakis, M.~Nieszporski, V.~Papageorgiou, and A.~Tongas.
Tetrahedron maps and symmetries of three dimensional integrable discrete equations.
\emph{J. Math. Phys.} \textbf{60} (2019) {123503}.

\bibitem{Sokor-Sasha}
S. Konstantinou-Rizos and A.V. Mikhailov.
Darboux transformations, finite reduction groups and related Yang--Baxter maps.
\emph{J. Phys. A: Math. Theor.} {\textbf{46}} (2013) {425201}.

\bibitem{Sharygin-Talalaev}
I.G. Korepanov, G.I. Sharygin, and D.V. Talalaev.
Cohomologies of $n$-simplex relations.
\emph{Math. Proc. Cambridge} \textbf{161} (2016), {203--222}.

\bibitem{Nijhoff}
J.M. Maillet and F. Nijhoff. The tetrahedron equation and the four-simplex equation.
\emph{Phys. Lett. A} \textbf{134} (1989), {221--228}.

\bibitem{Sergeev}
S.M. Sergeev. Solutions of the Functional Tetrahedron Equation Connected 
with the Local Yang--Baxter Equation for the Ferro-Electric Condition.
\emph{Lett. Math. Phys.} \textbf{45} (1998), {113--119}.

\bibitem{TalalUMN21}
D.V.~Talalaev. Tetrahedron equation: algebra, topology, and integrability.
\emph{Russian Math. Surveys} \textbf{76} (2021), 685--721.

\bibitem{vesel1991}
A.P.~Veselov. Integrable maps.
\emph{Russian Math. Surveys} \textbf{46} (1991), 1--51.

\bibitem{Yoneyama21}
A.~Yoneyama.
Boundary from Bulk Integrability in Three Dimensions: 3D Reflection Maps from Tetrahedron Maps. 
\emph{Math. Phys. Anal. Geom.} \textbf{24} (2021), 21.

\bibitem{Zamolodchikov} 
A.B. Zamolodchikov.
Tetrahedra equations and integrable systems in three-dimensional space.
\emph{Sov. Phys. JETP} \textbf{52} (1980), 325--336.

\bibitem{Zamolodchikov-2} 
A.B. Zamolodchikov.
Tetrahedron equations and the relativistic $S$-matrix of straight strings in (2+1)-dimensions. 
\emph{Commun. Math. Phys.} \textbf{79} (1981), 489--505.

\end{thebibliography}
\end{document}